\documentclass[11pt]{article}
\usepackage[margin=1in]{geometry}

\pdfoutput=1

\usepackage{booktabs}
\usepackage{amsmath, amsthm, amssymb}
\usepackage{algorithmicx}
\usepackage{mathtools}
\usepackage{comment} 
\usepackage{color-edits} 
\usepackage[most]{tcolorbox}
\usepackage{thm-restate}
\usepackage{xfrac}
\usepackage{hyperref}
\usepackage{multirow}
\usepackage{caption}
\usepackage{bm}
\usepackage{newfloat}
\usepackage{enumitem}
\usepackage{xpatch}
\usepackage[utf8]{inputenc}
\usepackage[T1]{fontenc}

\allowdisplaybreaks

\def\polylog{\operatorname{polylog}}

\newcommand{\MPC}[0]{\ensuremath{\mathsf{MPC}}}

\newcommand{\twocycle}[0]{\textsc{2-Cycle}}
\newcommand{\AMPC}[0]{\ensuremath{\mathsf{AMPC}}}
\newcommand{\PRAM}[0]{\ensuremath{\mathsf{PRAM}}}
\newcommand{\EREW}[0]{\ensuremath{\mathsf{EREW}}}
\newcommand{\CREW}[0]{\ensuremath{\mathsf{CREW}}}

\newcommand{\MultiPRAM}[0]{MultiPrefix \ensuremath{\mathsf{PRAM}}}

\newcommand{\myparagraph}[1]{\vspace{0.2cm}\noindent {\bf #1}}

\newcommand{\increasedegree}[0]{\textsc{IncreaseDegrees}}

\newcommand{\listranking}[0]{\textsc{ListRanking}}
\newcommand{\shrink}[0]{\textsc{Shrink}}

\newcommand{\cycleconn}[0]{\textsc{CycleConnectivity}}
\newcommand{\connectivity}[0]{\textsc{Connectivity}}
\newcommand{\spanningforest}[0]{\textsc{SpanningForest}}

\newcommand{\rootforest}[0]{\textsc{RootForest}}
\newcommand{\biconnectivity}[0]{\textsc{BC-Labeling}}
\newcommand{\preordernumber}[0]{\textsc{PreorderNumber}}
\newcommand{\msfincreasedegree}[0]{\textsc{MSFIncreaseDegree}}
\newcommand{\msf}[0]{\textsc{MinimumSpanningForest}}

\DeclareMathOperator{\poly}{poly}

\newcommand{\E}[0]{\ensuremath{\mathbb{E}}}

\addauthor{sb}{blue} 
\addauthor{ld}{red} 
\addauthor{kl}{purple} 
\addauthor{he}{gray} 
\addauthor{ws}{orange} 
\addauthor{vm}{green} 

\newtheorem{theorem}{Theorem}
\newtheorem{lemma}{Lemma}[section]
\newtheorem{proposition}[lemma]{Proposition}
\newtheorem{corollary}[lemma]{Corollary}

\newtheorem{claim}[lemma]{Claim}

\definecolor{mygreen}{RGB}{20,140,80}
\definecolor{linkcolor}{RGB}{100,0,0}
\definecolor{mylightgray}{RGB}{230,230,230}
\definecolor{verylightgray}{RGB}{245,245,245}

\hypersetup{
     colorlinks=true,
     citecolor= mygreen,
     linkcolor= linkcolor,
     urlcolor= mygreen
}

\newcommand{\etal}[0]{\textit{et al.}}

\algnewcommand{\IIf}[1]{\State\algorithmicif\ #1\ \algorithmicthen}
\algnewcommand{\EndIIf}{\unskip\ \algorithmicend\ \algorithmicif}

\newcounter{myalgctr}

\newtcolorbox{OuterBox}[1][]{%
    breakable,
    enhanced,
    frame hidden,
    interior hidden,
    left=-5pt,
    right=-5pt,
    top=-5pt,
    float=p,
    boxsep=0pt,
    arc=0pt
#1}%

\newtcolorbox{InnerBox}[1][]{%
    enforce breakable,
    enhanced,
    colback=gray,
    colframe=white,
#1}%

\newenvironment{tbox}{
\vspace{0.2cm}
\begin{tcolorbox}[
                  enhanced,
                  boxsep=2pt,
                  left=1pt,
                  right=1pt,
                  top=4pt,
                  boxrule=1pt,
                  arc=0pt,
                  colback=white,
                  colframe=black,
	              breakable
                  ]
}{
\end{tcolorbox}
}

\newcommand{\tboxhrule}[0]{\vspace{0.1cm} {\color{black} \hrule} \vspace{0.2cm}}

\newenvironment{titledtbox}[1]{\begin{tbox}#1 \tboxhrule}{\end{tbox}}

\newenvironment{tboxalg}[1]{\refstepcounter{myalgctr}\begin{titledtbox}{\textbf{Algorithm \themyalgctr.} #1}}{\end{titledtbox}}

\begin{document}\sloppy

\title{Massively Parallel Computation via Remote  Memory Access}

\author{
  Soheil Behnezhad\\University of Maryland
  \and
  Laxman Dhulipala\\CMU
  \and
  Hossein Esfandiari\\Google Research
  \and
  Jakub Łącki\\Google Research
  \and
  Warren Schudy\\Google Research
  \and
  Vahab Mirrokni\\Google Research
}

\date{}

%
%

\maketitle{}

\begin{abstract}
\setlength{\parskip}{0.4em}
We introduce the {\em Adaptive Massively Parallel Computation} (\AMPC{}) model, which is an extension of the {\em Massively Parallel Computation} (\MPC{}) model. At a high level, the \AMPC{} model strengthens the \MPC{} model by storing all messages sent within a round in a distributed data store. In the following round, all machines are provided with random read access to the data store, subject to the same constraints on the total amount of communication as in the \MPC{} model. Our model is inspired by the previous empirical studies of distributed graph algorithms~\cite{cc-beyond,nips17} using MapReduce and a distributed hash table service~\cite{bigtablepaper}.

This extension allows us to give new graph algorithms with much lower round complexities compared to the best known solutions in the \MPC{} model. In particular, in the \AMPC{} model we show how to solve maximal independent set in $O(1)$ rounds and connectivity/minimum spanning tree in $O(\log \log_{m/n} n)$ rounds both using $O(n^\delta)$ space per machine for constant $\delta < 1$. In the same memory regime for \MPC{}, the best known algorithms for these problems require $\poly\log n$ rounds. Our results imply that the \textsc{2-Cycle} conjecture, which is widely believed to hold in the \MPC{} model, does not hold in the \AMPC{} model.
\end{abstract}

\thispagestyle{empty}

\clearpage

\setcounter{page}{1}

\section{Introduction}\label{sec:intro}

The MapReduce computation framework was introduced in 2004 by Dean and Ghemawat~\cite{mapreduce}. It provides an easy to use platform for distributed computing, which handles synchronization and fault tolerance entirely on the framework side. Since its introduction MapReduce has inspired a number of other distributed computation frameworks, for example Spark~\cite{spark}, Hadoop~\cite{hadoop}, FlumeJava~\cite{flumejava}, Beam~\cite{beam}, Pregel~\cite{pregel}, and Giraph~\cite{giraph}.

These frameworks share many similarities, especially from a theoretical point of view. For this reason, the {\em Massively Parallel Computation} (\MPC{}) model~\cite{DBLP:journals/jacm/BeameKS17,DBLP:conf/soda/KarloffSV10,DBLP:conf/isaac/GoodrichSZ11} is widely accepted as the standard theoretical model of such systems.
Due to the generality of the \MPC{} model, it does not account for specific features used by emerging data-processing systems, such as RDMAs and RPCs whose use may result in simpler and faster algorithms.
In particular, in the context of graph algorithms, the major shortcoming of the \MPC{} model is the fact that exploring the neighborhood of a vertex is costly. The widely believed \textsc{2-Cycle} conjecture~\cite{DBLP:conf/icml/YaroslavtsevV18, DBLP:conf/spaa/RoughgardenVW16, cc-contractions, DBLP:journals/corr/abs-1805-02974} states that distinguishing between a graph being a cycle of length $n$ from a graph consisting of two cycles of length $n/2$ requires $\Omega(\log n)$ rounds. Intuitively, finding even a single vertex at a distance $k$ from a given vertex seems to require $\Omega(\log k)$ rounds. This restriction turns out to be, perhaps, {\em the} most important bottleneck in designing efficient \MPC{} algorithms for graph problems when the space per machine is much smaller than $n$, the number of vertices. In practice, however, the local neighborhood of a vertex can be explored much more efficiently. In a broader context, a data stored on a remote machine can be read with only a few microsecond latency using hardware support for remote direct-memory access (RDMA) without requiring a synchronized round of communication.

In this work, we attempt to close this gap between theory and practice by introducing an extension of the \MPC{} model and exploring its capabilities.

\subsection{Our contribution}
In this paper we introduce the \emph{Adaptive Massively Parallel Computation} model (\AMPC{}), which is an extension of the \MPC{} model. Roughly speaking, the \AMPC{} model extends the \MPC{} model by allowing machines to access a shared read-only memory {\em within a round}. We model this by assuming that all messages sent in a single round are written to a distributed data storage, which all machines can read from within the next round.

We demonstrate the usefulness of the new model by giving new algorithms for a number of fundamental graph problems. The corresponding round complexities are significantly lower than the complexities of the best known algorithms in the $\MPC{}$ model. Our results are summarized in Figure~\ref{tab:results}. We highlight that our results imply that the \textsc{2-Cycle} conjecture does not hold in the \AMPC{} model. Hence, under the conjecture, the \AMPC{} model is strictly stronger than \MPC{}.

\begin{figure*}[h]
  \centering
\begin{tabular}{l|c|c}
    \toprule
    \textbf{Problem} & \textbf{\AMPC{}} & \textbf{\MPC{}} \\
    \midrule
      Connectivity & $O(\log \log_{m/n} n)$ & $O(\log D \cdot \log \log_{m/n} n)$~\cite{andoniparallel}\\
      Minimum spanning tree & $O(\log \log_{m/n} n)$ & $O(\log n)$\\
      2-edge connectivity & $O(\log \log_{m/n} n)$ & $O(\log D \cdot \log \log_{m/n} n)$~\cite{andoniparallel}\\
      Maximal independent set & $O(1)$ & $\widetilde{O}(\sqrt{\log n})$~\cite{DBLP:conf/soda/GhaffariU19}\\
      \textsc{2-Cycle} & $O(1)$ & $O(\log n)$\\
      Forest Connectivity & $O(1)$ & $O(\log D \cdot \log \log_{m/n} n)$~\cite{andoniparallel}\\
    \bottomrule
  \end{tabular}
\caption{\label{tab:results}Round complexities of our algorithms in the \AMPC{} model compared to the state-of-the-art \MPC{} algorithms. $D$ denotes the diameter of the input graph.
	We consider the setting where the space per machine is sublinear in the number of vertices, that is $S = O(n^{\epsilon})$ for constant $\epsilon < 1$.}
\end{figure*}



Our model is inspired by two previous papers that combine the use of distributed hash table service (DHT) and MapReduce~\cite{cc-beyond, nips17}.
The first paper~\cite{cc-beyond} considers the problem of finding connected components.
It shows that by using a DHT in addition to MapReduce, the bound on the number of rounds can be decreased from $O(\log^2 n)$ to $O(\log n \log \log n)$.
It also demonstrates empirically that using a DHT can significantly decrease the number of rounds and the running time of the algorithm.
The second paper~\cite{nips17} studies the distributed hierarchical clustering problem, and shows how the use of a DHT facilitates the affinity clustering algorithm to scale up to graphs with trillions of edges.

However, neither of these papers defined a formal model of computation in which their algorithms operate.
Our paper is thus the first attempt to model the use of a distributed hash table along with a MapReduce-like distributed processing framework.

\subsection{Organization of the Paper}
In Section~\ref{sec:model} we describe our model and discuss its properties.
In particular, we discuss various practical aspects of implementing our new model.
Then, in Section~\ref{sec:prelims} we introduce notation that is used in the following sections.
The remaining sections present our algorithms.
We start with a very simple algorithm for solving the \textsc{2-Cycle} problem in Section~\ref{sec:twocycle}.
This is followed by algorithms for maximal independent set (Section~\ref{sec:mis}), connectivity (Section~\ref{sec:connectivity}), minimum spanning forest (Section~\ref{sec:msf}) and forest connectivity (Section~\ref{sec:forest-connectivity}).

\section{The Adaptive Massively Parallel Computation (\AMPC{}) Model}\label{sec:model}
\newcommand{\dds}{\mathcal{D}}

In the \AMPC{} model, we have an input of size $N$\footnote{The input size $N$ should not be confused with the number of vertices in a graph, which we denote by $n$.}, which is processed by $P$ machines, each with space $S$. We denote by $T$ the total space of all machines, i.e., $T = S \cdot P$.
We assume that $S = \Theta(N^{1-\Omega(1)})$ and require that the total space is close to the input size, i.e., $T = O(N \polylog N)$.

In addition, in the \AMPC{} model there is collection of distributed data stores (DDS) that we denote by $\dds_0, \dds_1, \dds_2, \ldots$. For the description, it is convenient to assume that all DDS-es provide key-value store semantics, that they store a collection of key-value pairs: given a key a DDS returns the corresponding value. We require each key-value pair to have constant size (i.e. both key and value consist of a constant number of words). The input data is stored in $\dds_0$ and uses a set of keys known to all machines (e.g. consecutive integers).

If $k > 1$ key-value pairs stored in the DDS have the same key $x$, the individual values can be accessed by querying for keys $(x, 1), \ldots, (x, k)$. Note that the indices from $1$ to $k$ are assigned arbitrarily. Querying for a key that does not occur in the DDS, results in an empty response.

The computation consists of \emph{rounds}. In the i-th round, each machine can read data from $\dds_{i-1}$ and write to $\dds_i$. Within a round, each machine can make up to $O(S)$ reads (henceforth called \emph{queries}) and $O(S)$ \emph{writes}, and perform arbitrary computation\footnote{The algorithms described in this paper perform near-linear computation in each step (w.r.t. their available space). However, we do not put a bound on the amount of computation on a machine, as this is a usual assumption of the \MPC{} model.}. Each query and write refers to querying for/writing a single key-value pair.
The amount of \emph{communication} that a machine performs per round is equal to the total number of queries and writes.

The key property of the model is that the queries that a machine makes in a round may depend on the results of the previous queries it made in the same round. This is why we call our model \emph{adaptive}.
In particular, if $g$ is a function from $X$ to $X$ and for each $x \in X$, $\dds_{i-1}$ stores a key-value pair $(x, g(x))$, then in round $i$ a machine can compute $g^k(y)$ in a single round, provided that $k = O(S)$.

It is easy to simulate every \MPC{} algorithm in the \AMPC{} model. Namely, in the \AMPC{} model, instead of sending a message to machine with id $x$, we can write a key-value pair keyed by $x$ to the DDS. In the following round, each machine reads all key-value pairs keyed by its id. The limits on the number of queries and writes in the \AMPC{} model are direct counterparts of the communication limits of the \MPC{} model.

Due to known simulations of \PRAM{} algorithms by \MPC{}~\cite{DBLP:conf/soda/KarloffSV10, DBLP:conf/isaac/GoodrichSZ11}, the \AMPC{} model can also simulate existing \PRAM{} algorithms from the \EREW{}, \CREW{}, \CREW{} models~\cite{Blelloch:2010:PA:1882723.1882748}, as well as stronger \PRAM{} variants such as the \MultiPRAM{}~\cite{Ranade:1989:FPC:916125}. The simulations can be done using $O(1)$ rounds per \PRAM{} step, and total space proportional to the number of processors.

\subsection{Realism}\label{sec:realism}
Let us now discuss the practical aspects of building a framework implementing the \AMPC{} model.

\myparagraph{Sequential queries.}
The \AMPC{} model allows each machine to issue $O(S)$ queries within each round. Each query may depend on the results of the previous queries, which may cause the queries to be issued sequentially. It is natural to wonder how this aspect of the model may affect the parallelism. We believe that this assumption is realistic for the following reasons.

First, recent work on remote direct memory-access (RDMA) has shown
that large-scale RDMA systems can be implemented efficiently, with
many commercial systems, hardware and networking support for RDMA
such as Infiniband~\cite{infiniband}, Omnipath~\cite{omnipath,
birrittella15omnipath} and RoCE~\cite{roce} being widely available. RDMA has
been used in recent work on key-value storage
systems~\cite{mitchell2013rdmakvstores, kalia14kvstore, lim2014mica,
dragojevic2014farm}, database systems~\cite{zamanian2017end}, and
other related systems and studies~\cite{ousterhout2015ramcloud,
dragojevic2017rdma}.  The latency for a remote memory access in recent
RDMA systems is around 1--3 microseconds under load, making it about
$20x$ more costly than a main-memory access~\cite{dragojevic2017rdma}.

Second, we observe that the per-machine latency can be further reduced by using the \emph{parallel slackness} technique~\cite{valiant1990bridging}. Assume there are $P = n^{\epsilon}$ machines, where $0 < \epsilon < 1$. Instead of running the algorithm directly on the machines, we split each machine into $T^\delta$ virtual machines of space $S = T^{\epsilon - \delta}$, for some small constant $0 < \delta < \epsilon$. Then, we execute the algorithm on the virtual machines.
Each physical machine can simulate the virtual machines concurrently. In particular, when a virtual machine is stalled on a network access, the physical machine can swap it out and run a step of the next virtual machine, thereby hiding the latency of the query (note the purpose of hyper-threading in modern systems is also to perform latency-hiding).

Observe that this requires that the algorithm that is executed can be run on machines with space $S = T^{\epsilon - \delta}$. We remark that most of the $\MPC{}$ algorithms that require $S = T^{1-\Theta(1)}$ space per machine can run as long as $S = T^\epsilon$, for any constant $\epsilon$ (incurring only a constant factor increase in the number of rounds). In particular, this property holds for all algorithms presented in this paper.


\myparagraph{Fault tolerance.}
One of the key properties of the popular distributed computation frameworks implementing the \MPC{} model is fault-tolerance. This is achieved by storing the messages sent to each machine on persistent storage, usually a distributed file system~\cite{mapreduce, DBLP:conf/mss/ShvachkoKRC10,spark, pregel}. When a machine fails, it can simply restart its computation from scratch and read its input again. The \AMPC{} model has similar characteristics. The key requirement is to implemented each DDS in a fault-tolerant way~\cite{bigtablepaper}. After that, a failing machine can be simply replaced with a different machine that would perform the computation from scratch. Note that this crucially relies on the fact that the contents of $\dds_{i-1}$ do not change within round $i$.
Overall, we believe that the \AMPC{} model can be implemented with similar fault-tolerance characteristics to the implementations of the \MPC{} model.

\myparagraph{Query contention.}
In the \AMPC{} model each machine can issue $O(S)$ queries to the DDS in each round, but the machines are otherwise not restricted in accessing the DDS. In particular, all machines may be accessing the same key-value pair, which could potentially lead to contention. We observe that under the following natural assumptions, this should not be an issue.
\begin{enumerate}[topsep=5pt,itemsep=2pt]
	\item $P = O(S^{1 - \delta})$ for a constant $\delta > 0$. This is a realistic condition for at least two reasons. First, typically the number of machines $P$ is up to few thousands, but the space $S$ of each machine is billions of bytes of memory. Apart from that, in the popular distributed computation frameworks, there is a master machine that communicates with all other machines in each round~\cite{mapreduce, pregel}. To make sure that the amount of this communication is bounded by $O(S)$, we must have $P = O(S)$.
\item The DDS is handled by $P$ machines, each having $O(S)$ space. (In practice these can be either the machines doing the computation, or a different set of machines.) Each of these machines is capable of answering at most $O(S)$ queries, which is in line with our assumptions on the machines performing the computation.
\item The key-value pairs are randomly and independently assigned to the machines handling the DDS. The choice of the keys queried by each machine is independent of this mapping. This is a natural assumption, which basically says that the algorithm executed by a machine is independent of the technical details of the system.
\item Each worker machine queries for each key at most once. Note that this can be assumed without loss of generality, since the machines have sufficient space to cache the results of the queries they make.
\end{enumerate}

The lemma below uses the above assumptions to analyze contention. It shows that each machine handling the DDS needs to answer $O(S)$ queries w.h.p.
Note that in the statement of the lemma the balls are the key-value pairs and their weights give the number of times they are queried for.

\begin{restatable}{lemma}{contention}\label{lem:contention}
Consider $T$ balls, each having an integer weight between $0$ and $P$.
The total weight of all balls is $T$.
Assume each of the balls is inserted uniformly at random to one of $P$ bins. The bin for each ball is chosen independently of the choices for other balls and of the ball weight.
Let $S = T/P$ and assume that $P=O(S^{1-\Omega(1)})$.

Then, the total weight of balls in each bin is $O(S)$ with high probability.
\end{restatable}

\myparagraph{Disallowing writes.}
In principle, the \AMPC{} model could allow not only random reads, but also writes within a round. However, this would have a number of undesirable consequences. First of all, failure recovery becomes much more complex. When a machine fails, it cannot simply restart its computation within a round from scratch, since the contents of the DDS may have changed since the round has started. Moreover, if writes were allowed, the machines could wait for the result of the computation of other machines, but the model would not account for this waiting time. In fact, one round of \AMPC{} with writes could be used to simulate $O(S)$ steps of a PRAM algorithm.

\section{Preliminaries}\label{sec:prelims}
We denote an unweighted graph by $G(V, E)$, where $V$ is the set of
vertices and $E$ is the set of edges in the graph. The number of
vertices in a graph is $n = |V|$, and the number of edges is $m =
|E|$. Vertices are assumed to be indexed from $0$ to $n-1$.  For
undirected graphs we use $N(v)$ to denote the neighbors of vertex $v$
and $\emph{deg}(v)$ to denote its degree. We use $D$ to refer to the
diameter of the graph, or the longest shortest path distance between
any vertex $s$ and any vertex $v$ reachable from $s$. We assume that
there are no self-edges or duplicate edges in the graph.

Most of the algorithms presented in this paper take a graph as an input.
In this case, the input size $N$ is equal to $n + m$, that is the total number of vertices and edges in the graph.
We also assume that $S = n^\epsilon$, where $\epsilon$ is a constant, such that $\epsilon \in (0, 1)$.
Depending on the algorithm, the total space of all machines required by the algorithm is either $\Theta(N)$ or $\Theta(N \log N$).
In the latter case, the number of machines is $P = \Theta((N \log N) / n^\epsilon)$.
In some cases, we may obtain a faster algorithm if $T = n^{1+\Omega(1)}$.

We say that an event happens \emph{with high probability}
(w.h.p.) if it happens with probability at least $1- 1/n^{c}$, where the constant $c$ can be made arbitrarily large.

In the descriptions of our algorithms we highlight the parts which rely on \AMPC{} model features.
The remaining parts can be implemented in \MPC{} model using standard primitives, such as sorting, duplicate removal, etc.


\section{Warm-up: The \twocycle{} Problem}\label{sec:twocycle}

In this section we give an algorithm that given an instance of the \twocycle{}
problem, solves it in $O(1)$ rounds w.h.p. Recall that an instance of
the problem is either a single cycle on $n$ vertices, or two cycles on
$n/2$ vertices each. The problem is to detect whether the graph
consists of one cycle, or two cycles. Our algorithm for the problem
runs over a number of rounds. In each round, we sample each vertex
with probability $n^{-\epsilon/2}$. Our goal is to replace the paths
between sampled vertices with single edges, which can be viewed as
contracting the original graph to the samples. We perform this
contraction by traversing the cycle in both directions for each
sampled vertex until we hit another vertex that is sampled. Traversing
the cycle is implemented using the adaptivity of the model. In order to achieve this, it suffices to store in the DDS the graph, in which the sampled vertices are marked.
Using a sample probability of $n^{-\epsilon/2}$, we can show that the number
of vertices shrinks by a factor of $n^{\epsilon/2}$ in each round
w.h.p. Therefore, after $O(1/\epsilon)$ rounds, the number of
remaining vertices and edges is reduced to $O(n^{\epsilon})$, at which
point, we can fit the graph in the memory of a single machine and
solve the remaining problem in a single round.

\begin{lemma}\label{lem:iteration-shrink}
Let $G$ be a graph consisting of cycles, and let $n$ be the initial
number of vertices in $G$. Consider a cycle with size
$k = \Omega(n^{\epsilon})$ in some iteration of the loop of
$\shrink{}(G, \epsilon, O(1/\epsilon))$. The size of this cycle
shrinks by at least a factor of $n^{\epsilon/2}$ after this iteration w.h.p.
\end{lemma}

\begin{tboxalg}{$\shrink(G=(V, E), \delta, t)$}\label{alg:shrink}
	\begin{enumerate}[ref={\arabic*-}, topsep=0pt,itemsep=0ex,partopsep=0ex,parsep=1ex, leftmargin=*]
    \item For $i := 1, \ldots, t$:
			\begin{enumerate}[ref={(\alph*)}, topsep=5pt,itemsep=0ex]
        \item Sample each vertex independently with probability $n^{-\delta/2}$.
		Denote the set of sampled vertices by $M$.
          Randomly distribute the vertices of $M$ to the machines.

        \item For each sampled vertex $v$ traverse the cycle in each direction
	until a sampled vertex is reached. Let $l_v$ and $r_v$ be the sampled vertices,
					where the each of two traversals finishes.

      \end{enumerate}
    \item Return the graph $(M, \{xl_x \mid x \in M\} \cup \{xr_x \mid x \in M\})$.
  \end{enumerate}
\end{tboxalg}

\begin{tboxalg}{$\twocycle(G=(V, E))$}\label{alg:twocycle}
	\begin{enumerate}[ref={\arabic*-}, topsep=0pt,itemsep=0ex,partopsep=0ex,parsep=1ex, leftmargin=*]
    \item $G' := \shrink{}(G, \epsilon, O(1/\epsilon))$
    \item Solve the \twocycle{} problem on $G'$, which has size
      $O(n^{\epsilon})$ w.h.p. on a single machine.
	\end{enumerate}
\end{tboxalg}

\begin{proof}
Consider the cycle, which by assumption has size $\Omega(n^{\epsilon})$.
Let the vertices in this cycle be labeled $1, \ldots, k$, and let $X_i$
be the random variable that is $1$ if the $i$th vertex is marked, and
$0$ otherwise. Let $X = \sum_{i=1}^{k} X_i$. Since the sampling
probability in each round is $n^{ - \epsilon/2}$, we have that $E[X]
= \frac{k}{n^{\epsilon/2}} = \Omega(\log n)$. Applying a Chernoff
bound, we have:
\[
  P(X \geq (1 + c)E[X]) \leq \exp(-c^{2} \log n / 3) = n^{-c^{2}/3}.
\]
The number of vertices in this cycle in the next iteration is just the
number of cycle vertices that are sampled. By the argument above, the
number of sampled vertices is $O(E[X]) = O(\frac{k}{n^{\epsilon/2}})$ w.h.p.
	which is a shrink by a factor of $\Omega(n^{\epsilon/2})$.
\end{proof}

\begin{lemma}\label{lem:twocyclerounds}
Let $G$ be an $n$-vertex graph consisting of cycles and let $G' = \shrink{}(G, \epsilon, O(1/\epsilon))$.
Then $G'$ is a graph which can be obtained from $G$ by contracting edges.
	The length of each cycle in $G'$ is $O(n^\epsilon)$.
\end{lemma}
\begin{proof}
From the pseudocode, it is clear that $G'$ can be obtained from $G$ by contracting edges.
By Lemma~\ref{lem:iteration-shrink}, the size of a cycle with size $k =
\Omega(n^{\epsilon})$ shrinks by a factor of $n^{\epsilon/2}$ w.h.p.
after an iteration of the loop in \shrink{}.
%
Hence, after the first iteration, the length of each
cycle(s) is $O(n^{1 - \epsilon/2})$ w.h.p, and after the $i$th
iteration, the length of the cycle(s) is $O(n^{1 - i\epsilon/2})$.
After $\frac{2(1-\epsilon)}{\epsilon}$ iterations, the size of the
cycle(s) is $O(n^{\epsilon})$.
\end{proof}

\begin{lemma}\label{lem:twocyclecomm}
The total communication of each machine is $O(n^{\epsilon})$ in each
round w.h.p., where $n$ is the initial number of vertices in $G$.
\end{lemma}
\begin{proof}
Following the proof of Lemma~\ref{lem:twocyclerounds}, in the $i$th
iteration, there are $O(n^{1 - i(\epsilon/2)})$ vertices remaining, and
$O(n^{1 - (i+1)\epsilon/2})$ vertices sampled w.h.p. Recall that there are
$P = \Omega(n^{1-\epsilon})$ machines, so each machine receives
$O(n^{1 - (i+1)(\epsilon/2) - (1 - \epsilon)}) = O(n^{\epsilon/2})$
samples w.h.p.

Observe that the length of each traversal is a geometric random variable with $p =
n^{-\epsilon/2}$. Let $Z_i$ be the number of vertices that the $i$th
sample traverses until it hits a sampled vertex. Let $Z =
\sum_{i=0}^{n^{\epsilon/2}} Z_i$ be a random variable measuring the
total amount of queries used by the samples assigned to some machine.
Using a Chernoff bound for geometric random variables, and the fact
that $E[Z] = n^{\epsilon}$, for some arbitrary large constant $c$ we have
\[
  P(Z \geq (1 + c)E[Z]) \leq \exp(-c^{2} E[Z] / 3) \leq \exp(-c^{2}
  \log n) = O(n^{-c}).
\]
Therefore, the total number of queries sent by the machine is
$O(n^{\epsilon})$ on any round w.h.p.
\end{proof}

Combining the previous lemmas gives us the following theorem:
\begin{theorem}\label{thm:twocycle}
  There is an \AMPC{} algorithm solving the \twocycle{} problem in
  $O(1/\epsilon)$ rounds w.h.p. using $O(n^\epsilon)$ space per
  machine and $O(n)$ total space.
\end{theorem}

\section{Maximal Independent Set}\label{sec:mis}

In this section, we prove the following result.

\begin{theorem}\label{thm:mis}
	For any constant $\epsilon \in (0, 1)$, there exists an \AMPC{} algorithm that finds a maximal independent set of a given graph $G(V, E)$ in $O(1)$ rounds. The algorithm is randomized and fails with arbitrarily small constant probability, uses $O(n^\epsilon)$ space per machine and $O(m)$ total space.
\end{theorem}

\newcommand{\lfmis}[1]{\ensuremath{\mathsf{LFMIS}(#1)}}

Our algorithm for Theorem~\ref{thm:mis} finds the {\em lexicographically first MIS} over a random permutation of the vertices. Given a permutation $\pi: [n] \to [n]$, the lexicographically first MIS according to $\pi$, which we denote by $\lfmis{G, \pi}$, is constructed by processing the vertices in the order of $\pi$ and adding a vertex $v$ to the MIS if and only if it does not lead to a conflict. That is, $v$ joins the MIS if none of its already processed neighbors have joined it before.

This ``greedy''-type construction of $\lfmis{G, \pi}$ is not directly applicable in our model if we require a small number of rounds. Instead, we have to parallelize this algorithm. For this, we first recall a query process by \cite{yoshida, DBLP:conf/focs/NguyenO08} to determine whether a given vertex $v$ belongs to $\lfmis{G, \pi}$, then show how we can use this query process to efficiently construct $\lfmis{G, \pi}$ in $O(1)$ rounds.

For a vertex $v$, define $f(v, \pi)$ to be the indicator function for whether $v \in \lfmis{G, \pi}$. That is, $f(v, \pi) = 1$ if $v \in \lfmis{G, \pi}$ and $f(v, \pi)=0$ otherwise. Note that we hide the dependence on $G$ in the $f(v, \pi)$ notation for brevity. Determining the value of $f(v, \pi)$ can be done in the following recursive way.

\begin{tboxalg}{$f(v, \pi)$}\label{alg:misoracle}
	\begin{enumerate}[topsep=0pt,itemsep=0ex,partopsep=0ex,parsep=1ex, leftmargin=*]
		\item Let $u_1, \ldots, u_d$ be the neighbors of $v$ sorted such that we have $\pi(u_1) < \ldots < \pi(u_d)$.
		\item For $i = 1, \ldots, d$:
			\begin{itemize}[topsep=0pt,itemsep=0ex,partopsep=0ex,parsep=1ex, leftmargin=*]
				\item If $\pi(u_i) < \pi(v)$:
				\begin{itemize}
					\item if $f(u_i, \pi) = 1$, then return $f(v, \pi) = 0$.
				\end{itemize}
			\end{itemize}
		\item Return $f(v, \pi)=1$.
	\end{enumerate}
\end{tboxalg}

Now, for a vertex $v$, define $q_\pi(v)$ to be the number of (recursive) calls to the algorithm above for determining the value of $f(v, \pi)$. A result of Yoshida et al.~\cite{yoshida} can be succinctly summarized as:

\begin{proposition}\label{pr:yoshida}
	If permutation $\pi$ is chosen uniformly at random, then $\E_\pi[\sum_{v \in V} q_\pi(v)] \leq m+n$.
\end{proposition}

Suppose that we fix a random permutation over the vertices, say, by each vertex $v$ picking a random real $\rho_v \in [0, 1]$ and constructing $\pi$ by sorting the vertices based on $\rho$. Having this, the natural way of implementing Algorithm~\ref{alg:misoracle} in the \AMPC{} model is to assign each vertex $v$ to a specific machine and simulate $f(v, \pi)$ there by exploring the relevant neighborhood of $v$ according to the query process. The main challenge is that while Proposition~\ref{pr:yoshida} bounds the total communication by $O(m)$ in expectation, there is no upper bound on $q_\pi(v)$ for a given vertex $v$. Note that the machine responsible for vertex $v$ has to query $q_\pi(v)$ vertices. Since the number of such queries within a round is bounded by the local space on the machines ($S$), this idea does not work if $q_v = \omega(n^\epsilon)$ for some vertex $v$. To ensure that we do not violate the per machine capacity on the number of queries and the total communication size, we use the iterative and truncated variant of the algorithm shown in Algorithm~\ref{alg:mpcmis}.
\newcommand{\mis}[0]{\textsc{MaximalIndependentSet}}

\newcommand{\isinmis}[1]{\ensuremath{\mathsf{is\text{-}in\text{-}MIS}(#1)}}
\newcommand{\unknown}[0]{\ensuremath{\mathsf{unknown}}}
\newcommand{\true}[0]{\ensuremath{\mathsf{true}}}
\newcommand{\false}[0]{\ensuremath{\mathsf{false}}}
\newcommand{\truncatedquery}[0]{\ensuremath{\textsc{TruncatedQuery}}}

It is clear from the description that $\isinmis{v}$ either equals the value of $f(v, \pi)$ or is $\unknown$ and that once Algorithm~\ref{alg:mpcmis} terminates it returns $\lfmis{G, \pi}$. It remains to analyze the number of required iterations.

\begin{tboxalg}{$\mis(G=(V, E))$}\label{alg:mpcmis}
	\begin{enumerate}[topsep=0pt,itemsep=0ex,partopsep=0ex,parsep=1ex, leftmargin=*]
		\item Fix a permutation $\pi: [n] \to [n]$ on the vertices uniformly at random.
		\item Each vertex $v$ is assigned to a machine $\mu_v$ chosen uniformly at random from $m/n^\epsilon$ machines.
		\item For each vertex $v$, set $\isinmis{v} \gets \unknown$.
		\item Repeat while the graph is not empty:\label{line:loop}
			\begin{enumerate}[topsep=0pt,itemsep=0ex,partopsep=0ex,parsep=1ex, leftmargin=*]
				\item For each vertex $v$ remaining in the graph, $\mu_v$ runs $\truncatedquery(v, \pi, n^\epsilon)$. If $\isinmis{v} = \true$, then for every neighbor $u$ of $v$, set $\isinmis{u} \gets \false$ and remove $u$ from the graph as well.
			\end{enumerate}
	\end{enumerate}
\end{tboxalg}

\begin{tboxalg}{$\truncatedquery(v, \pi, c)$}\label{alg:capacitatedmis}
	\textbf{Input:} a vertex $v$, a permutation $\pi$ and a capacity $c$ on the number of (recursive) queries allowed.
	
	\textbf{Output:} the number of conducted queries.
	\begin{enumerate}[topsep=10pt,itemsep=0ex,partopsep=0ex,parsep=1ex, leftmargin=*]
		\item If $c=0$, then return 0.
		\item Let $u_1, \ldots, u_d$ be the remaining neighbors of $v$ sorted such that $\pi(u_1) < \ldots < \pi(u_d)$.
		\item $q \gets 1$
		\item For $i = 1, \ldots, d$:
			\begin{enumerate}[topsep=0pt,itemsep=0ex,partopsep=0ex,parsep=1ex, leftmargin=*]
				\item If $\pi(v) < \pi(u_i)$, set $\isinmis{v} = \true$ and return $q$.
				\item $q \gets q + \truncatedquery(u_i, \pi, c-q)$
				\item If as a result of the query above, we get that $u_i$ is in the MIS, then set $\isinmis{v} = \false$ and return $q$.
				\item If $q \geq c$, then return $q$.
			\end{enumerate}
	\end{enumerate}
\end{tboxalg}



As we will later discuss, each iteration of the Line~\ref{line:loop} loop in Algorithm~\ref{alg:mpcmis} can be implemented in $O(1)$ rounds of $\AMPC$. Thus, it suffices to bound the number of iterations of this loop to bound the round complexity of the algorithm. In what follows, we simply say ``iteration'' to refer to iterations of this loop. We use the following lemma to bound the number of iterations by $O(1/\epsilon)$. For simplicity, we say a vertex $v$ {\em is settled} if $\isinmis{v} \not= \unknown$.

\begin{lemma}\label{lem:iteration}
	By the end of iteration $i$, any vertex $v$ with $q_\pi(v) \leq n^{i \epsilon/2}$ will be settled.
\end{lemma}

We prove Lemma~\ref{lem:iteration} by induction. For the base case with $i=1$, it is clear that by the end of iteration 1, any vertex with $q_\pi(v) \leq n^{\epsilon/2}$ will be settled since the capacity on the number of queries is $n^\epsilon$. Suppose that this holds for any iteration $i \leq r-1$. We show that by the end of iteration $r$, any vertex $v$ with $q_v \leq n^{r\epsilon/2}$ will be settled.

For any vertex $v$, define $s_i(v) = \langle w_1, w_2, \ldots, w_q \rangle$ to be the sequence of vertices for which function $\truncatedquery$ is (recursively) called within iteration $i$ sorted in the same order that they were visited. Observe that for a vertex $v$ that is not settled by the end of iteration $i$,  $s_i(v)$ has to have exactly $n^\epsilon$ vertices since this is the cap that we set for any vertex in Algorithm~\ref{alg:mpcmis}. The following claim will be useful in proving Lemma~\ref{lem:iteration}.

\begin{claim}\label{cl:hitsettledhighquery}
	Fix an arbitrary vertex $v$ and let $D$ be an upper bound on the depth of recursions within $\truncatedquery(v, \pi, c)$. If $v$ is not settled by the end of iteration $r$,  the following holds if we denote $s_r(v) = \langle w_1, \ldots, w_q \rangle$: for any integer $j \in [q-D]$, there is at least one vertex $w \in \{w_j, w_{j+1}, \ldots, w_{j+D}\}$ such that $q_\pi(w) \geq n^{(r-1)\epsilon/2}$ and that $w$ is settled during iteration $r$ without calling $\truncatedquery$ on any other vertex.
\end{claim}

\begin{proof}
	Since $D$ is an upper bound on the depth of recursions, there is at least one vertex $w_{k+1} \in \{w_j, \ldots, w_{j+D}\}$ for $k \geq 1$ that is not queried by $w_{k}$ (otherwise the query depth from $w_j$ to $w_{j+D}$ would be $D+1$ which is a contradiction). We show that vertex $w_{k}$ satisfies the requirements.
	
	Observe that in Algorithm~\ref{alg:capacitatedmis}, we call $\truncatedquery$ only on vertices that are not already settled. Since $w_k \in s_r(v)$, vertex $w_k$ should not have been settled by the end of iteration $r-1$. By the induction hypothesis, this means that $q_\pi(w_k) > n^{(r-1)\epsilon/2}$. Furthermore, the fact that vertex $w_{k+1}$, the successor of $w_k$ in $s_r(v)$, is not queried by $w_k$, means that $w_k$ has not queried any vertices and thus it should be settled within round $r$, proving the claim.
\end{proof}

Now, take vertex $v$ with $q_v \leq n^{r\epsilon/2}$ and suppose for contradiction that $v$ is not settled by round $r$. Recall that this means $|s_r(v)| = \Theta(n^\epsilon)$. By Claim~\ref{cl:hitsettledhighquery}, there are at least $\Omega(|s_r(v)|/D)$ vertices $u$ with $q_u \geq n^{(r-1)\epsilon/2}$ that are settled within round $r$ without querying any vertex. This saves $v$ from querying $\Omega(\frac{|s_r(v)|}{D} \cdot n^{(r-1)\epsilon/2}) = \Omega(\frac{n^\epsilon}{D} \cdot n^{(r-1)\epsilon/2})$ vertices  compared to the original query process $f(v, \pi)$. The recursion depth can be bounded in a straightforward way by adapting the arguments of \cite{DBLP:conf/spaa/BlellochFS12} (see the paragraph below) by $D = O(\log^2 n)$. Thus, a total of $\Omega(\frac{n^\epsilon}{\log^2 n} \cdot n^{(r-1)\epsilon/2}) \gg n^{r\epsilon/2}$ queries should not be sufficient to settle the status of $v$. This is a contradiction since we assumed $q_v \leq n^{r\epsilon/2}$.

\myparagraph{Bounding the recursion depth.} The analysis of \cite{DBLP:conf/spaa/BlellochFS12} bounds the round complexity of a parallel implementation of the greedy (lexicographically first) MIS algorithm by $O(\log^2 n)$. While applying this result as a black-box does not bound the depth of recursions, here we argue that with minor modifications, it also bounds the recursion depth by $O(\log^2 n)$. The analysis of \cite{DBLP:conf/spaa/BlellochFS12} divides the vertex set into disjoint layers $V_1, V_2, \ldots, V_k$ where $k=O(\log n)$ such that the vertices in $V_i$ have higher priority of joining MIS than all those in $V_j$ for $j > i$. The key property is that after finding the vertices of the lexicographically first MIS among $V_1, \ldots, V_i$ for any $i$, and removing their neighbors, the remaining vertices in $V_{i+1}$ have logarithmic degree within their partition and thus, any path going through vertices of increasing priorities has length $O(\log n)$. Collecting all these paths of all layers, this gives a $O(\log^2 n)$ round complexity for the parallel implementation of the algorithm.

In order to adapt this to our setting, suppose that the recursion depth of a vertex is larger than $\beta\log^2 n$ for some large enough constant $\beta$. This means that in one of the layers defined by \cite{DBLP:conf/spaa/BlellochFS12}, say $V_i$, the query process has to visit more than $\alpha \log n$ vertices for some sufficiently large constant $\alpha$. Note that by the analysis of \cite{DBLP:conf/spaa/BlellochFS12}, there should be a vertex in layers $V_1, \ldots, V_{i-1}$ that destroys this path by arguments of \cite{DBLP:conf/spaa/BlellochFS12}. Let $v$ be a vertex that has to be destroyed. This means that $v$ has to have a neighbor in lower layers which joins the MIS.  The crucial observation, here, is that since the query process is greedy, $v$ will immediately go to lower levels and will be eventually deleted from the graph before visiting any other vertex in layer $V_i$. Thus, the arguments of \cite{DBLP:conf/spaa/BlellochFS12} still hold and the recursion depth is bounded by $O(\log^2 n)$.

\section{Undirected Graph Connectivity}\label{sec:connectivity}
In this section we provide an implementation of the recent undirected connectivity
algorithm due to Andoni \etal{}~\cite{andoniparallel} in the \AMPC{} model in  $O(\log\log_{T/n} n)$ rounds.
This improves upon the original result by a factor of $\log D$, where $D$ is the diameter of the input graph.
Note that by our assumption $T = O(n + m)$.
However, if we had $T = n^{1+\Omega(1)}$, our algorithm would run in a constant number of rounds.

At a high level, we show that using the \AMPC{} model, we can speed-up
the central procedure in the algorithm of Andoni \etal{}, so that it
runs in only $O(1)$ rounds, compared to $O(\log D)$ rounds in the
original algorithm.  However, for completeness we provide the entire
algorithm.

\myparagraph{The Andoni \etal{} Algorithm.}
If the total amount of space used is not a concern, undirected
connectivity can easily be solved in $O(\log D)$ rounds in the \MPC{}
model with sublinear space per machine as follows.
Let $G^2$ be a \emph{squared} graph, that is a graph obtained by adding edges between each vertex and all its $2$-hop neighbors.
By repeating the squaring $\lceil \log_2 D\rceil$ times, we obtain a graph, where 
each connected component is transformed into a clique.  The issue with the squaring approach is that the total amount of space used to store the resulting graph
can be much as much as $\Theta(n^{2})$, which is prohibitive for sparse graphs.

The main idea behind the Andoni \etal{} algorithm is to implement the
graph exponentiation idea in a more space-efficient manner.
The algorithm runs over a number of phases and in each phase contracts the graph, reducing the number of vertices.
The algorithm also assigns each vertex a \emph{budget} $d$, making sure that the sum of budgets of all vertices is $O(T)$.
As the number of vertices decreases, the budgets are increased.
The budgets of all vertices are equal to each other at all times.

Let us now describe one phase.
The first step of a phase increases the degree of each vertex to at least $d$, by adding edges within each connected component (and without adding parallel edges).
After that every vertex is sampled as a {\em leader}
independently with probability $\Theta(\log n/d)$ (assume $d = \Omega(\log n)$ for
simplicity).
The non-leader vertices are then contracted to leaders (a leader exists in each
non-leader's neighborhood w.h.p.). The total number of vertices after
contraction
is $\tilde{O}(n/d)$. Therefore, in the next round, the
budgets can be increased to $d^{1.4}$ and the sum of budgets will be
$O(n d^{0.4}) = O(T)$. As the budgets increase double exponentially, $O(\log
\log n)$ phases are required before each connected component shrinks to one vertex.

In the Andoni \etal{} algorithm, each phase is implemented by using
$O(\log D)$ rounds of squaring. In each round, each vertex that has
degree less than its budget connects itself to its 2-hop (or a
sufficiently sized subset of its 2-hop if the 2-hop is too large). In
the worst case, the algorithm requires $O(\log D)$ rounds to make the
degrees equal to the budget.

\myparagraph{Implementing Andoni \etal{} in \AMPC{}.}
We show how to implement each phase in $O(1)$ rounds using the adaptivity in the
\AMPC{} model. Recall that the goal of a phase is to increase the
degree of each vertex in the current graph to $d$, the current
per-vertex budget.
In the \AMPC{} model, this can be achieved by exploring the local
neighborhood of each vertex in one round (e.g. using BFS).
The search stops once it visits $d$ vertices.

One possible worry is that there may be significant overlap amongst
the neighbors that we query while probing from a vertex. Namely,
running BFS until it visits $d$ vertices may require visiting $\Omega(d^2)$ edges.
We cope with this issue by setting the budget to be the square root of the space available per each vertex.
Therefore, when a vertex attempts to increase its degree to $d$ in a phase, we can afford to issue $d^2$ queries.

Once $d$ reaches $n^{\epsilon/2}$ we no longer increase it, which assures that the number of queries issued per vertex is at most $O(n^\epsilon)$.
From this point the number of vertices decreases by a factor of $d = \Omega(n^{\epsilon/2})$ in each iteration.
It follows that after $O(1/\epsilon)$ iterations each
connected component will be contracted to a single vertex. The \AMPC{} algorithm for a phase is given in Algorithm~\ref{alg:increase-degree}.

\vspace{-0.2cm}

\begin{tboxalg}{$\increasedegree(G=(V, E), d)$}\label{alg:increase-degree}
	\begin{enumerate}[ref={\arabic*-}, topsep=0pt,itemsep=0ex,partopsep=0ex,parsep=1ex, leftmargin=*]
    \item For each vertex $v$, run a BFS from $v$,
      until it visits $d$ vertices or visits the entire connected
      component of $v$. Let $N_v$ be the set of visited vertices.

\item Return $(V, E \cup \{vx \mid v \in V, x \in N_v\})$
	\end{enumerate}
\end{tboxalg}

Algorithm~\ref{alg:increase-degree} can be naturally implemented in the \AMPC{} model.
First, the graph is written to the DDS.
Then, the vertices are randomly assigned to machines, and each machine runs the BFS for its assigned vertices.
We now analyze this implementation.

\begin{lemma}\label{lem:increase-degree}
Algorithm~\ref{alg:increase-degree} runs in $O(1)$ rounds of \AMPC{}.
The total number of queries is $O(nd^2)$, where $n$ denotes the number of vertices of $G$.
If the number of machines is $n^{1-\epsilon}d^2$, the maximum number of queries a machine makes is $O(n^{\epsilon})$ w.h.p.
\end{lemma}
\begin{proof}
The algorithm clearly runs in $O(1)$ rounds. To bound the number of
queries, note that in the worst case we make a
query for every edge between a pair of vertices that are visited by
the BFS. Since we stop after visiting $d$ vertices, the maximum number
of queries we make is $O(d^2)$ and therefore the total number of
queries is $O(n^{\epsilon}d^2)$.

Next we use a Chernoff bound to bound the maximum number of queries a machine makes. Fix a machine $i$ and let binary random variable $X_v$ indicate whether vertex $v$ is assigned to machine $i$. Let $X=\sum_{v\in V} X_v$. Note that $X$ is the total number vertices assigned to machine $i$, and $E[X]= \frac{n^{\epsilon}}{d^2}\geq n^{\epsilon/3} $. By a Chernoff bound we have 
\begin{align*}
P\big(X\geq 2 E[X] \big)\leq \exp{\left(-\frac{E[X]}{3}\right)} = \exp{(-n^{\epsilon/3}/3)}.
\end{align*}
A union bound over all $d^2 n^{1-\epsilon}$ machines implies that the total number of queries of each machine is upper bounded by $E[X] \cdot d^2 = O(n^{\epsilon})$, with high probability. 
\end{proof}

By plugging in the above procedure to the algorithm of Andoni \etal{},
we are able to reduce the running time by a factor of $\log D$.

In addition, we also show how to ensure that the algorithm runs
correctly with high probability, as opposed to the constant
probability of success of the original algorithm.  In order to achieve
this, we need to assure that after each vertex is chosen as a leader
with probability $\Theta(\log n / d)$, every vertex has a leader among
its neighbors with high probability.

To this end, we need to assure that the degree of each vertex in the
graph before sampling is $\Omega(\log n)$.  We can achieve this by
running Algorithm~\ref{alg:increase-degree} with $d = \Omega(\log n)$.
By Lemma~\ref{lem:increase-degree}, the number of queries issues by
the algorithm is $nd^2$, which does not exceed the allowed number of
queries $T = O(m)$, whenever $m = \Omega(n \log^2 n)$.  Hence, we
handle the case when $m = o(n \log^2 n)$ separately.

%


When $m = o(n \log^2 n)$ we use the algorithm of~\cite{manuscript} to shrink the number
of vertices by a factor of $\Omega(\log^2 n)$ in $O(\log\log_{T/n} n)$ rounds w.h.p.
This procedure is implementable in the \MPC{} model and hence implementable in the \AMPC{} model as well.
The following lemma summarizes the properties of the algorithm.

\begin{lemma}[\cite{manuscript}]\label{lem:smalltotalspace}
 There exists an $O(\log \log n)$ round \MPC{} algorithm using $O(n^\epsilon)$
  space per machine and $O(m)$ total space that with high probability,
  converts any graph $G(V, E)$ with $n$ vertices and $m$ edges to a
 graph $G'(V', E')$ and outputs a function $f: V \to V'$ such that:
	\begin{enumerate}[ref={\arabic*},itemsep=-0.5ex]
		\item The number of non-isolated vertices in $G'$ is $O(n / \log^2 n)$ .
		\item $|E'| \leq |E|$.
		\item For any two vertices $u$ and $v$ in $V$, vertices $f(u)$ and $f(v)$ in $V'$ are in the same component of $G'$ if and only if $u$ and $v$ are in the same component of $G$.
	\end{enumerate}
\end{lemma}


\myparagraph{Undirected Connectivity.}
We assume that $m = \Omega(n\log^2 n)$ in what follows.
Next, we  describe the undirected connectivity algorithm in the \AMPC{} model.

\clearpage

\begin{tboxalg}{\connectivity{}(G=(V, E))}\label{alg:connectivity}
	\begin{enumerate}[ref={\arabic*-}, topsep=0pt,itemsep=0ex,partopsep=0ex,parsep=1ex, leftmargin=*]


		\item Let $d = \sqrt{T/n}$. Let $M
      : V \rightarrow \mathbb{Z}$ be a mapping from vertices to a
      unique identifier for their component (initially mapping each
      vertex to its own id).

    \item While $G$ contains at least one edge, perform the following steps:

    \begin{enumerate}
      \item Randomly assign the vertices in $G$ to machine and call
        $\increasedegree{}(G, d)$ for each vertex. Add the edges
        found to $G$.\label{step:increasedegree}

      \item Sample each vertex with probability $\Theta(\log n
        / d)$ to be a leader. \label{step:sampleleaders}

      \item Let $G'$ to be the graph formed by contracting each vertex
        $v$ to a leader in its neighborhood. If its
        degree is less than $d$, then contract it to the neighbor with
        lowest id.  Let $M'$ be a mapping $V(G) \rightarrow
        V(G')$ mapping each vertex to the vertex it contracts to in
        $G'$. Finally, update $G$ to be $G'$.
        \label{step:contractgraph}

\item Update $d = \min(d^{1.4}, n^{\epsilon/3}$).

      \item Update $M$ by mapping each $v \in V$ to $M'(M(v))$.
        \label{step:updatemapping}

    \end{enumerate}
    \item Output $M$
	\end{enumerate}
\end{tboxalg}

\begin{theorem}\label{thm:connectivity-rounds-queries}
Algorithm~\ref{alg:connectivity} computes the connected components of
an undirected graph in $O(\log \log_{T/n} n + 1/\epsilon)$ rounds of \AMPC{} w.h.p.
where the total space $T = \Omega(m + n)$.
\end{theorem}
\begin{proof}
For simplicity let us assume that the graph is connected.
Otherwise, we can apply the analysis to each connected component.
First, we show that the sum of all \emph{squared} budgets is $O(T)$ at every step of the algorithm.
	Indeed, this is true in the beginning, and in each iteration the number of nodes decreases by a factor of $d / \log n$ w.h.p. while the per-vertex squared budgets are increased by a factor of $d^{0.8}$.
Clearly, the algorithm also assures that $d = O(n^{\epsilon/2})$.
Since we assumed that $m = \Omega(n \log^2 n)$, in the first iteration we have $d = \Omega(\log n)$.
This implies that in this iteration, as well as all following ones, calling \increasedegree{} will increase degrees to at least $d$, except for the vertices contained in connected components of size at most $d$.
Hence, w.h.p. each vertex either has a neighbor who is a leader, or belongs to a connected component of size at most $d$, which \increasedegree{} turns into a clique.

From the analysis of~\cite{andoniparallel} it follows that the number of iterations of the algorithm until $d = \Omega(n^{\epsilon/3})$ is $O(\log \log_{T/n} n)$.
Once $d$ reaches $\Omega(n^{\epsilon/3})$, the number of vertices in each connected component shrinks by a factor of $\Omega(n^{\epsilon/3})$ in each iteration.
	Hence, the total number of iterations is $O(\log \log_{T/n} n + 1/\epsilon)$.
Each iteration can be implemented in $O(1)$ rounds (see Lemma~\ref{lem:increase-degree}.

It remains to analyze the number of queries in each round.
Clearly, it suffices to analyze the number of queries for the part of the algorithm that cannot be implemented in the \MPC{} model, which is Algorithm~\ref{alg:increase-degree}.
By Lemma~\ref{lem:increase-degree} and the fact that the sum of squared budgets is $O(T)$, we get that the total number of queries in one round is $O(T) = O(m)$.
Moreover, the maximum number of queries per machine is $O(n^{\epsilon})$ w.h.p.
\end{proof}

\section{Minimum Spanning Forest}\label{sec:msf}
In this section we describe a minimum spanning forest (MSF) algorithm
with a round-complexity of $O(\log\log_{T/n} n)$ w.h.p. in the \AMPC{}
model.  The input to the algorithm consists of an undirected weighted
graph, $G = (V, E)$.  For simplicity, we assume that all edge weights
are distinct to ensure that there is a unique MSF (the assumption can
easily be met by, for example, breaking ties by the ids of the
endpoints). The algorithm is similar in spirit to the connectivity
algorithm in Section~\ref{sec:connectivity}. In the first phase of the
algorithm, each vertex computes a subset of the minimum spanning
forest using Prim's algorithm~\cite{CLRS}. In the second phase, we
randomly sample leaders like in the connectivity algorithm, and
contract each non-leader vertex to a leader in its neighborhood.

When $m = o(n \log^2 n)$ we use the algorithm of Lemma~\ref{lem:smalltotalspace}, which shrinks the number of vertices by a factor of $\Omega(\log^2 n)$ in $O(\log\log_{T/n} n)$ rounds w.h.p.
The procedure executes a number of steps, each of which shrinks the number of vertices by a constant factor w.h.p.
We invoke each step of the procedure with the graph formed by each vertex and the lowest weight edge incident to the vertex (like in Bor\r{u}vka's
algorithm) so that the contractions performed will be along edges in
the MSF. Since the procedure decreases the number of vertices by a
factor of $\Omega(\log^2 n)$ in $O(\log\log_{T / n} n)$ rounds w.h.p.,
in what follows we assume that $m = \Omega(n\log^2 n)$.

Algorithm~\ref{alg:msf-increase-degree} gives pseudocode for the
degree increasing procedure in our minimum spanning forest algorithm.
It takes as input a graph $G$, and a degree-bound $d$, and runs Prim's
algorithm from each vertex, $v$, until the size of the MSF rooted at
$v$ is $d$.

\begin{tboxalg}{$\msfincreasedegree(G=(V, E), d)$}\label{alg:msf-increase-degree}
	\begin{enumerate}[ref={\arabic*-}, topsep=0pt,itemsep=0ex,partopsep=0ex,parsep=1ex, leftmargin=*]
    \item For each vertex, $v$, assigned to a machine, compute its
    local spanning forest, $F_{v}$, and $E(v)$ as follows:
    \begin{enumerate}
      \item Initially, $F_{v} = \{v\}, E(v) = \{\}$
      \item While $|F_{v}| < d$, select the minimum weight edge,
      $(v_1, v_2)$ s.t. $v_1 \in F_{v}$ and $v_2 \notin F_{v}$, and
      add $v_{2}$ to $F_{v}$. Add $(v_1, v_2)$ to $E(v)$. If such an
      edge does not exist, halt, since $F_v$ contains $v$'s entire
      component.
    \end{enumerate}
    \item Output $F_{v}$ and $E(v)$ for each vertex $v$.
	\end{enumerate}
\end{tboxalg}

\begin{tboxalg}{$\msf(G=(V, E))$}\label{alg:msf}
	\begin{enumerate}[ref={\arabic*-}, topsep=0pt,itemsep=0ex,partopsep=0ex,parsep=1ex, leftmargin=*]

    \item Let $E_{MSF} = \emptyset$ and $M = \{w \rightarrow e\ |\ e=(u, v, w) \in E\}$ be a
      mapping from edge weights to edges.

    \item Let $G_{C} = G, d = \sqrt{T/n}$

    \item While $G_{C}$ is non-empty, perform the following steps:

    \begin{enumerate}
      \item Randomly assign the vertices in $G_{C}$ to machines, and
        call $\msfincreasedegree{}(G_{C}, d)$ for each vertex. The
        output is a set of neighbors, $F_{v}$ and a set of minimum
        spanning forest edges, $E(v)$ for each $v \in V$.
        \label{mst:findtreeedges}

      \item $E_{MSF} = E_{MSF} \cup \{M[E(V)]\ |\ v \in V\}$.
        Note that we apply $M$ to the MSF edges found in this round in
        order to add the original edge in $G$ corresponding to some
        edge $G_{C}$.

      \item Sample each vertex with probability $\Theta(\log n
        / d)$ to be a leader. \label{mst:sampleleaders}

      \item Update $G_{C}$ to the graph formed by contracting each
        vertex $v$ to a leader in $F_{v}$. If its degree is less than
        $d$, then contract it to the neighbor in $F_{v}$ with lowest
        id. Update $M$ on the new graph.

      \item Update $d = \min(d^{1.4}, n^{\epsilon/3}$).

    \end{enumerate}
    \item Output $E_{MSF}$.
	\end{enumerate}
\end{tboxalg}

Algorithm~\ref{alg:msf} is the main MSF routine. We first decrease
the number of vertices in the graph if $m = o(n\log^2 n)$. Next, we
call Algorithm~\ref{alg:msf-increase-degree} starting with an initial
setting of $d = \Omega(\log n)$. We then perform leader selection to
reduce the number of vertices in the graph by a factor of $d$. The
decrease in the number of vertices allows us to increase $d$ to
$d^{1.4}$ in the next phase.  Therefore, after $O(\log \log_{T/n} n)$
such phases, $d$ becomes $O(n^{\epsilon/3})$, at which point, the
algorithm terminates after $O(1/\epsilon)$ rounds.

\begin{lemma}
Algorithm~\ref{alg:msf} correctly outputs the minimum spanning forest
of $G$.
\end{lemma}
\begin{proof}
  The proof is by induction on the number of iterations, basing on the
  last iteration. In the base case, note that the iteration finds all
  remaining MSF edges, since $G_{C}$ becomes empty after the final
  iteration by assumption. Since MSF edges are discovered using Prim's
  algorithm, and the edge weights are distinct, this iteration
  correctly computes the MSF of the graph.

  In the inductive case, we just need to show that we find a valid
  subset of the MSF edges in this iteration, and that any MSF edge not
  yet found is still present in $G_{C}$. Again, since we discover
  edges using Prim's algorithm and due to each edge having a unique
  edge weight, we commit a valid subset of the MSF edges in this
  iteration. To show that any MSF edge not yet found is still
  preserved in $G_{C}$, observe that an edge $(u,v)$ is only removed
  from $G_{C}$ if both $u$ and $v$ merge to the same leader.

  Suppose that an MSF edge $(u,v)$ that is not committed is removed,
  i.e., is not represented in $G_{C}$. Then by the argument above,
  $(u,v)$ merge to the same leader, and there is therefore a path $P$
  composed of MSF edges committed in this round going between $u$,
  $l(u)$ and $v$. Since we assumed that $(u,v)$ is an MSF edge, some
  edge in $P$ must have larger weight than $(u,v)$, but this is a
  contradiction, as it implies that the local searches using Prim's
  algorithm are incorrect. Thus, any removed edge cannot be an MSF
  edge, and so all MSF edges that are not found belong to
  $G_{C}$.
\end{proof}

\begin{theorem}\label{thm:msf-rounds-queries}
Algorithm~\ref{alg:msf} computes the minimum spanning forest of an
undirected graph in $O(\log \log_{T/n})$ rounds of \AMPC{} w.h.p.
where the total space $T = \Omega(m + n)$.
\end{theorem}
\begin{proof}
  The proof follows by the same argument as the proof of
  Theorem~\ref{thm:connectivity-rounds-queries}.
\end{proof}

Since we can compute a spanning forest of an undirected graph by
assigning arbitrary distinct weights to the edges, we have the
following corollary:

\begin{corollary}\label{cor:spanningforest-rounds-queries}
A spanning forest of an undirected graph can be found in $O(\log
\log_{T/n})$ rounds of \AMPC{} w.h.p.
where the total space $T = \Omega(m + n)$.
\end{corollary}

\section{Forest Connectivity}\label{sec:forest-connectivity}

In this section, we present algorithms for solving \emph{forest
  connectivity}, i.e., the undirected graph connectivity problem on
forests in the \AMPC{} model. 
Our forest connectivity algorithm is based on the classic
technique of transforming each forest into a cycle, via an Eulerian tour. After representing the forest as a set of Eulerian tours,
finding the connectivity of the forest reduces to solving connectivity
on a collection of cycles.
Observing that the PRAM construction of Tarjan~\etal{}~\cite{tarjan1985efficient} can be implemented in \AMPC{} implies that this reduction to
the cycle connectivity problem can be done in $O(1)$ rounds of
\AMPC{}.

Next, we address the cycle connectivity problem.

\myparagraph{Cycle connectivity.}
The \emph{cycle connectivity problem} is to compute the connected
components of a graph $G=(V, E)$ which contains a set of disjoint
cycles.
Our cycle connectivity algorithm, Algorithm~\ref{alg:cycleconn}, works
as follows. We call $O(1/\epsilon)$ iterations of \shrink{} (the same
algorithm used in the \twocycle{} problem) which reduces the size of
the largest cycle to $O(n^{\epsilon})$ with high probability. Then, we switch
to a different algorithm which computes the connected components of
the remaining cycles in one round using a total of $O(n\log n)$
communication with high probability. The algorithm first fixes a
permutation $\pi$. Then, each vertex in the graph searches along
one direction of the cycle until it either returns to itself or hits a
vertex that appears before it in the permutation (i.e., has higher rank in
$\pi$).

\begin{tboxalg}{$\cycleconn(G=(V, E))$}\label{alg:cycleconn}
	\begin{enumerate}[ref={\arabic*-}, topsep=0pt,itemsep=0ex,partopsep=0ex,parsep=1ex, leftmargin=*]
    \item Let $G' := \shrink{}(G, \epsilon, O(1/\epsilon))$

    \item Fix a random permutation, $\pi$, over the vertices of $G'$
      and randomly assign the vertices to the machines.

    \item \label{step:search} For each vertex $u$ assigned to a
      machine, search one direction of the cycle until either $u$
      fully traverses the cycle or $u$ encounters a vertex $v$ s.t.
      $\pi(v) < \pi(u)$. The vertex with lowest rank in $\pi$ in a
      cycle is the representative for the cycle.
	\end{enumerate}
\end{tboxalg}

We obtain the following corollary by applying Lemma~\ref{lem:iteration-shrink}.
\begin{corollary}\label{cor:largestcycle}
  After $O(1/\epsilon)$ iterations of \shrink{} with $\delta =
  \epsilon/2$, the largest cycle in the graph has size
  $O(n^{\epsilon/2})$ with high probability.
\end{corollary}

\begin{lemma}\label{lem:cyclequeryexpec}
Consider a cycle of length $k = \Omega(\log n)$. The number of queries
made by a vertex until it hits a vertex with higher rank in $\pi$ is
$O(\log k)$ in expectation.
\end{lemma}
\begin{proof}
  Given an arbitrary vertex $v$, the probability that we have to query
  $i$ vertices before finding a vertex is equal to the probability that the
  $i-1$ vertices following $v$ in the cycle have lower rank than
  $v$ in the permutation, and that the $i$th vertex has higher rank
  than $v$ in the permutation. This probability is $\frac{1}{i} \cdot
  \frac{1}{i+1}$. The expected number of queries is therefore
  \[
    \sum_{i=1}^{k} i \cdot \frac{1}{i} \cdot \frac{1}{i+1} =
    \sum_{i=1}^{k} \frac{1}{i+1} = H_k - 1 \in O(\log k),
  \]
  completing the proof.
\end{proof}

Next, we show that the total number of queries for large enough
cycles is concentrated.
\begin{lemma}\label{lem:cycleconnquicksort}
In Algorithm~\ref{alg:cycleconn}, the total number of queries of a
cycle of length $k = \Omega(\log n)$ is $O(k \log k)$ with high
probability.
\end{lemma}
\begin{proof}
We observe that the amount of work done by the queries is identical to
the work done by pivots selected in the randomized quicksort
algorithm. To make the analogy concrete, consider first picking the
highest priority element in $\pi$, and charge it $k$ work for the
queries it performs.  Now, consider the second highest priority
element, which takes one of the remaining $k-1$ slots. This element
may perform $k-1$ queries in the worst case, but it breaks the
remaining array into two independent subproblems which cannot query
each other. The query complexity follows from the analysis of
randomized quicksort (see for example~\cite{motwani1995randomized}.).
\end{proof}

Next we bound the total number of queries of a machine. Since our
bound on the total query complexity $O(n\log n)$ w.h.p., we use
$O(n\log n)$ total space in order to perform the extra queries.
Therefore, in what follows, we assume that we have $O(n^{1-\epsilon}\log n)$ machines.
The proof boils down to analyzing a weighted version of
the balls and bins problem, when we have a upper-bound on the maximum
weight of any ball. Although a very similar problem is considered
in~\cite{sanders1996competitive}, it does not provide a high
probability bounds on the maximum load of a bin.

\begin{lemma}\label{lem:cycleconnwhp}
The total amount of queries performed by a machine in
Algorithm~\ref{alg:cycleconn} is $O(n^{\epsilon})$ w.h.p.
\end{lemma}
\begin{proof}
The number of queries performed by \shrink{} is at most
$O(n^{\epsilon})$ per round by Lemma~\ref{lem:twocyclecomm}. We now
show that the total queries performed per machine in step (3) when
traversing the cycles is $O(n^{\epsilon})$ w.h.p.

To see this, we apply a lemma concerning the mean of a sample of size
$k$ from a universe of size $n$ without replacement. Given a set of
real numbers $c_1, \ldots, c_n$ and the sample $X_1, \ldots, X_k$,
define $X = \sum_{i=1}^{k} X_i$, $\bar{c}_{n} = \frac{1}{n}
\sum_{i=1}^{n} c_i$, $\sigma^{2}_{n} = \frac{1}{n} \sum_{i=1}^{n} (c_i
- \bar{c}_{n})^{2}$ and lastly $\Delta_{n} = \max_{i=1}^{n} c_i -
\min_{i=1}^{n} c_i$.
The Hoeffding-Bernstein bound (see Lemma 2.14.19
from~\cite{van1997weak}) states that
\[
P(|X - \bar{c}_{n}| > t) \leq 2 \exp \bigg(- \frac{kt^{2}}{2\sigma^{2}_{n} + t \Delta_{n}} \bigg)
\]
Taking $c_i$ to be the query complexity of the $i$th vertex, we now
compute the terms defined above.
Notice that by
Lemma~\ref{lem:cycleconnquicksort} we have $\sum_{i=1}^{n} c_i =
O(n\log n)$ w.h.p., and so $\bar{c}_{n} = O(\log n)$, with high
probability.
Furthermore, by Corollary~\ref{cor:largestcycle} we have $c_i < n^{\epsilon/2}$, with high probability.
Therefore, we have $\forall i, (c_i - \bar{c}_{n})^{2} \leq n^{\epsilon}$,
and $\sigma^{2}_{n} \leq n^{\epsilon}$. Lastly, we have $\Delta_{n} \leq
n^{\epsilon/2}$. Letting $t = c \cdot n^{\epsilon}$, we have
\begin{align*}
P(|X - \bar{c}_{n}| > c n^{\epsilon})
    &\leq 2 \exp \bigg(- \frac{c^{2}n^{3\epsilon}}{2n^{\epsilon} + cn^{3\epsilon/2}} \bigg)\\
    &\leq 2 \exp(-cn^{\epsilon}) \\
    &\leq 2 n^{-c},
\end{align*}
completing the proof.
\end{proof}

Putting the lemmas above together, we have the following theorem:
\begin{theorem}\label{thm:forestconn}
  There exists an \AMPC{} algorithm that solves the forest
  connectivity problem in $O(1/\epsilon)$ rounds of computation w.h.p.
  using $T = O(n\log n)$ total space w.h.p.
\end{theorem}

\subsection{Rooting Trees and List Ranking}
We now consider the problem of computing a rooted tree given a set of
edges representing a tree. We also present algorithms for computing
tree properties, such as subtree sizes, preorder numbering, and
subtree minima/maxima in this section. These algorithms are used in
our 2-edge connectivity algorithm.

We start with the tree rooting problem. The input is a set of tree
edges, and a vertex $r$ to root the tree at. The output is a mapping
$p : V \rightarrow V$ from vertices in the tree to their
\emph{parents} where $p(r) = r$. A classic parallel algorithm for
rooting a tree in the is based on the Euler tour
technique~\cite{tarjan1985efficient}, which works as follows. First
compute an Euler tour of the tree. Rooting the tree reduces to the
\emph{list ranking} problem, a classic problem in the parallel
algorithms literature. We now show that a simple extension of our
algorithm for the \twocycle{} problem can be used to perform list
ranking in the \AMPC{} model.

The algorithm works as follows. It first computes the Euler tour of
the tree, and then breaks it at an edge incident to the root, $r$,
which unravels the cycle into a list. It then uses list ranking to
compute the distance between each edge and the root of the list. Since
each edge appears twice in the Euler tour, it also maintain
cross-pointers between these edges.

\begin{tboxalg}{$\listranking(L=\{v_0, \ldots,
    v_N\})$}\label{alg:list-ranking}
	\begin{enumerate}[ref={\arabic*}, topsep=0pt,itemsep=0ex,partopsep=0ex,parsep=1ex, leftmargin=*]
    \item Let $n = N$ and $r = 1$. Let $w_{r}(v)$ be the weight of
      vertex $v$ on the $r$th round. Initially, $w_{1}(v) = 1$.

    \item While $n = \Omega(N^{\epsilon})$, run the following
      steps.

      \label{lc:contractionloop}
			\begin{enumerate}[ref={(\alph*)}, topsep=5pt,itemsep=0ex]

        \item Sample each list element with probability
          $N^{-\delta/2}$ uniformly at random, and let $S$ be the set
          of samples, including $v_0$. Randomly distribute the samples
          to the machines.\label{lc:samplestep}

        \item For each sampled vertex, $v$, traverse the list until
          another sampled vertex is hit and let $l(v)$ be number of
          vertices traversed before hitting the next sample. Update
          $w_{r}(v)$ to be $w_{r-1}(v) + l(v)$.

        \item Contract the graph to the samples, and update $r = r + 1$.
      \end{enumerate}
    \item Solve the remaining weighted list-ranking problem on a
      single machine using $w_{r}$. The output is the correct rank of
      all vertices that survived all $r$ sampling rounds. Let
      $d_{r}(v) = w_{r}(v)$ for vertices that are active on the $r$th
      round and $0$ otherwise.

    \item For $r' = r-1 \ldots 0$, run the following steps:
			\begin{enumerate}[ref={(\alph*)}, topsep=5pt,itemsep=0ex]
        \item Let $v$ be a vertex sampled on round $r$ in
          Step~\ref{lc:samplestep}. Traverse the list until the next
          sampled vertex is hit, and update $d_{r'}(u_i) = d_{r'+1}(v) +
          \emph{dist}(v, u_i)$ for all vertices $v, u_1, \ldots,
          u_{l(v)}$.
      \end{enumerate}
    \item Output $d_{0}$.
  \end{enumerate}
\end{tboxalg}

\begin{lemma}\label{lem:listrankingcorrect}
Algorithm~\ref{alg:list-ranking} correctly solves the list ranking
problem.
\end{lemma}
\begin{proof}
First note that the contraction process of Algorithm~\ref{alg:list-ranking} does not change the order of the vertices. Moreover, for each vertex $v$, each vertex $u$ that is between $v$ and the root may only contract to another vertex in this range. Thus, if after a round a vertex $v$ survives, the total weight of the vertices between $v$ and the root remains the same. Next we use this observation to prove the statement of the lemma.

We prove this by an induction on the number of rounds, basing on the last round. For the base case, note that in after all contractions, we explicitly calculate the total weight of the vertices between each survived vertex and the root. As we mentioned before this is exactly the index of the survived vertex in the initial list. Next, assume that all the indices of the vertices in round $r'+1$ are correctly set. Note that these are the sampled vertices in round $r'$. Hence, to find the index of a (not sampled) vertex $v$ in round $r'$ we just need to sum up the the index of its next sampled vertex with the total weight of the vertices between $v$ and the sample vertex, as mentioned in the algorithm. This proves the induction step and completes the proof of the lemma.
\end{proof}

\begin{theorem}\label{thm:listranking}
  List ranking a list of length $N$ can be performed in
  $O(1/\epsilon)$ rounds of $\AMPC{}$ w.h.p. Furthermore, the
  algorithm uses $O(N)$ total space w.h.p.
\end{theorem}
\begin{proof}
A proof identical to that of Lemma~\ref{lem:iteration-shrink} shows
that a list with size $k = \Omega(N^{\epsilon})$ shrinks by a factor
of at least $N^{\epsilon/2}$ after one iteration of the loop in
Step~\ref{lc:contractionloop}. Therefore, $O(1/\epsilon)$ iterations
suffice before the length of the list is $O(N^{\epsilon})$ (the proof
is identical to that of Lemma~\ref{lem:twocyclerounds}). Bounding the
number of queries can be done similar to the proof of
Lemma~\ref{lem:twocyclecomm}, which implies the bound on the total
space.
\end{proof}

Observing that the PRAM construction of
Tarjan~\etal{}~\cite{tarjan1985efficient} can be implemented in \MPC{}
gives the following lemma.
\begin{lemma}\label{lem:etts}
Given a forest on $n$ vertices, $F$, as a collection of undirected
edges, the Euler tours for trees in $F$ can be constructed in
$O(1/\epsilon)$ rounds of \MPC{} using $O(n)$ total space.
\end{lemma}

Now, by applying forest connectivity, followed by the reduction from the
tree rooting problem to list ranking described above we have the
following theorem:
\begin{theorem}\label{thm:treerooting}
Given a forest on $n$ vertices, $F$, as a collection of undirected
edges, the trees in $F$ can be rooted in $O(1/\epsilon)$ rounds of
$\AMPC{}$ w.h.p using $O(n\log n)$ total
space w.h.p.
\end{theorem}

Note that list-ranking allows us to efficiently implement algorithms
based on the Euler tour of a tree.  The idea is to first compute an
Euler tour and apply list ranking from the root, which gives the
distance from each edge to the root. After list ranking, the Euler
tour can be represented as a sequence which we refer to as the
\emph{Euler sequence}, since for each edge its rank its position in an
array representing the tour starting from the root.  Computing
properties, such as subtree sizes, or preorder numberings follow
naturally from this representation of a tree.

\myparagraph{Subtree Sizes.}
Given a rooted tree, $T$, we show how to compute the subtree sizes for
each vertex in the tree. We first compute the Euler sequence of $T$,
starting at the root. Next we assign a weight of $1$ to each forward
edge and a weight of $0$ to each reverse edge (see the discussion on
rooting trees above). Computing a prefix-sum over the Euler sequence
with respect to the weights computes, for each edge, the number of
forward edges preceding it. The subtree size at a given node, $v$, can
be obtained by taking the difference of the values for the last
reverse edge entering $v$ and the first forward edge exiting $v$.
Therefore, we have the following lemma:
\begin{lemma}\label{lem:subtreesizes}
Computing the subtree sizes of a rooted tree can be performed in
$O(1/\epsilon)$ rounds of $\AMPC{}$ w.h.p. using $O(n)$ total space
w.h.p.
\end{lemma}

\myparagraph{Preorder Numbering.}
Computing a preorder numbering of a rooted tree can also be done using
the Euler sequence and prefix-sums. The algorithm stores a weight of
$1$ for each forward edge and a weight of $0$ for each reverse edge,
and compute the prefix-sum over the weighted Euler sequence with
respect to the weights. The preorder number of a vertex is simply the
number of forward edges preceding it.  Therefore, we have the
following lemma:
\begin{lemma}\label{lem:preordernumbering}
Computing the preorder numbering of a rooted tree can be performed in
$O(1/\epsilon)$ rounds of $\AMPC{}$ w.h.p. using $O(n)$ total space
w.h.p.
\end{lemma}

\myparagraph{Subtree Minimum and Maximum.}
Note that the Euler tour based technique described above to compute
subtree sizes is not sufficient to compute the minimum and maximum
over subtrees quickly (the technique depends on the function having an
inverse, which precludes using it for $\min, \max$). The classic
solution to this problem is the range minimum query (RMQ) data
structure. The data structure is built over an input array of $n$
elements. Queries take two indices, $i, j$ and return the minimum
value element in this range. Observe that after computing the Euler
sequence, computing the minimum and maximum of a vertex $v$'s subtree
can be reduced to a RMQ on the interval in the Euler tour representing
$v$'s subtree. Using the fact that RMQ can be implemented efficiently
in \MPC{} (e.g., see~\cite{andoniparallel}) gives us the following lemma:
\begin{lemma}\label{lem:subtreeminmax}
A data structure to compute the minimum and maximum of a subtree can
be computed in $O(1/\epsilon)$ rounds of \AMPC{}. Each query to the data structure can be answered in $O(1)$ rounds.
Preprocessing and storing the data structure uses $O(n)$ total space w.h.p.
\end{lemma}

\section{2-Edge Connectivity}\label{sec:biconnectivity}

In this section we show how our connectivity, spanning-forest
and tree-based algorithms can be combined to obtain a 2-edge
connectivity algorithm. A \emph{biconnected component} of an
undirected graph $G$ is a maximal subgraph such that it remains
connected even after removing any single vertex from it. A connected
graph can be decomposed into a block-cut tree, which contains
biconnected components, joined by vertices called \emph{articulation
  points}. A graph is \emph{$k$-edge connected} if it remains
connected under the removal of any $k-1$ edges. A \emph{bridge} is an
edge whose removal increases the number of connected components of
$G$. In this section we design a data structure that enables us

Sequentially, both problems can be solved using the Hopcroft-Tarjan
algorithm~\cite{hopcroft1973algorithm}. The algorithm uses depth-first
search (DFS) to identify the articulation points and bridges in $O(m +
n)$ time. Although this sequential algorithm can be parallelized using
parallel algorithms for DFS~\cite{aggarwal1989parallel}, the resulting
algorithm requires many processors and is unlikely to lead to an
efficient \MPC{} or \AMPC{} implementation. In the PRAM setting,
Tarjan and Vishkin present the first work-efficient algorithm for
biconnectivity and bridge-finding~\cite{tarjan1985efficient} (see also
Maon~\etal{}.~\cite{maon1986parallel} and
Ramachandran~\etal{}~\cite{ramachandran1992parallel} for
biconnectivity algorithms based on computing open ear decompositions).

\newcommand{\codevar}[1]{\mathit{#1}}

\begin{tboxalg}{\biconnectivity{}(G=(V, E))}\label{alg:biconnectivity}
	\begin{enumerate}[ref={\arabic*}, topsep=0pt,itemsep=0ex,partopsep=0ex,parsep=1ex, leftmargin=*]
    \item Let $(E_{F}, M) = \spanningforest{}(G)$. $E_{F}$ is the set
      of spanning forest edges and $M$ is the connectivity labeling of
      $G$. \label{bc:spanningforest}

    \item Let $F = \rootforest{}(E_{F})$ and $PN =
      \preordernumber{}(F)$. \label{bc:rootpreorder}

    \item For each $v \in V$, compute $\codevar{Low}(v)$ and
      $\codevar{High}(v)$ and $\codevar{Size}(v)$ where
      $\codevar{Low}(v)$ and $\codevar{High}(v)$ are the minimum and
      maximum preorder numbers of all non-tree $(u,w)$ edges where $u$
      is in $v$'s subtree, and where $\codevar{Size}(v)$ is the size
      of $v$'s subtree.\label{bc:lowhighsize}

    \item $\codevar{critical}$ = $\{(u, p(u)) \in F$ s.t.
      $(u, p(u))$ satisfies Equation~(\ref{eqn:bccheck})$\}$.
      \label{bc:critical}

    \item $L$ = \connectivity{}($G(V, E \setminus
      \codevar{critical})$).\label{bc:bridgeconn}

    \item Output $(L, F)$.
	\end{enumerate}
\end{tboxalg}

Algorithm~\ref{alg:biconnectivity} describes our \AMPC{}
implementation of the Tarjan-Vishkin
algorithm~\cite{tarjan1985efficient}. The algorithm first computes a
spanning forest of $G$ where the trees in the forest can be rooted
arbitrarily (Step~\ref{bc:spanningforest}). Next, it computes a rooted
forest $F$ by rooting each tree in the forest at an arbitrary vertex
and computes $PN$, a preorder numbering of each tree
(Step~\ref{bc:rootpreorder}). It then computes for each $v \in V$
three quantities: $\codevar{Low}(v), \codevar{High}(v),
\codevar{Size}(v)$ which are the minimum and
maximum preorder numbers respectively of all non-tree $(u, w)$ edges
where $u$ is a vertex in $v$'s subtree, and the size of each vertex's
subtree. The initial values for $\codevar{Low}(v), \codevar{High}(v)$,
and $\codevar{Size}(v)$ are $PN(v)$, $PN(v)$ and $0$ respectively.

The critical edges of $G$ are tree edges $(v, p(v)) \in F$ where
\begin{equation}\label{eqn:bccheck}
  \codevar{PN}(p(v)) \leq \codevar{Low}(v) \text{ and } \codevar{High}(v) \leq \codevar{PN}(p(v)) + \codevar{Size}(v)
\end{equation}
The algorithm computes a labeling that can then be used to retrieve
the bridges and articulation points in $O(1)$ queries as follows.
We first delete the critical edges from the graph and compute the
connectivity of this modified graph (Step~\ref{bc:critical}).
We follow Ben-David \etal{} and refer to this pair $(L, F)$ as the
\emph{BC-labeling} of the graph~\cite{bendavid2017implicit}.


Given the BC-labeling $(L, F)$, we can identify the bridges and
articulation points as follows.
The \emph{head} of a component is the parent of the component in the
spanning tree (the root is always a head). Each head defines a
distinct biconnected component (there are at most $n$ heads, since
there can be at most $n$ biconnected components).
A non-root vertex $v$ is an articulation point if it is the head of at
least one component. The root of the forest is an articulation point
if it is the head of two or more components, since these components
were not connected by a non-tree edge in the graph with critical edges
removed.
A tree edge $(u, p(u))$ is a bridge if $u$'s component in $L$ only
contains $u$.
Using this information we can compute 2-edge connected
components of the graph by removing all bridges and running
connectivity again.

Based on the description above, we have the following lemma about the
correctness and cost of our BC-labeling algorithm:
\begin{lemma}\label{lem:bclabelingcost}
  Given an undirected graph $G$, BC-labeling can be computed in
  $O(\log\log_{T/n} n)$ rounds of \AMPC{} w.h.p. using $O(T)$ total
  space w.h.p. where the total space $T = \Omega(m + n)$.
\end{lemma}
\begin{proof}
The correctness of this algorithm follows from the
proofs in Tarjan \etal{}~\cite{tarjan1985efficient}, and Ben-David
\etal{}~\cite{bendavid2017implicit}.

The cost is proved by analyzing the steps of the algorithm. The
spanning forest computation in Step~\ref{bc:spanningforest} runs in
$O(\log\log_{T/n} n)$ rounds of \AMPC{} w.h.p. by
Corollary~\ref{cor:spanningforest-rounds-queries}. Step~\ref{bc:rootpreorder} runs in
$O(1/\epsilon)$ rounds of \AMPC{} w.h.p. and uses $O(n\log n)$ space
w.h.p.by Theorem~\ref{thm:treerooting} and
Lemma~\ref{lem:preordernumbering}. Note that the space of the
tree rooting algorithm can be reduced to $O(n)$ by applying the
general graph connectivity algorithm
(Algorithm~\ref{alg:connectivity}) instead of the forest connectivity
algorithm, at the expense of $O(\log\log_{T}/n)$ rounds.
Step~\ref{bc:lowhighsize}, which computes $\codevar{Low},
\codevar{High}$, and $\codevar{Size}$ run in $O(1/\epsilon)$ rounds
due to Lemmas~\ref{lem:subtreesizes} and~\ref{lem:subtreeminmax}.
\end{proof}

Since we can use the BC-labeling to compute 2-edge connected
components of the graph in the same round-complexity, we have the
following theorem:
\begin{theorem}\label{thm:biconnectivity}
2-edge connectivity can be computed in $O(\log\log_{T/n} n)$ rounds of
\AMPC{} w.h.p. where the total space $T = \Omega(m + n)$.
\end{theorem}

\section{Conclusion and Open Problems}\label{sec:conclusion}

In this work, we introduced and justified the Adaptive Massively Parallel
Computation model, an extension of the Massively
Parallel Computation model. Furthermore, we presented several new
graph algorithms with much lower round complexities compared to the
best known algorithms in the \MPC{} model, such as a constant round
algorithm to find a maximal independent set and $O(\log \log_{m/n} n)$
round algorithms to find connectivity/minimum spanning tree among
others problems.
As future research, it is interesting to design algorithms for other
problems such as vertex coloring, edge coloring and maximal matching
in the \AMPC{} model. Another open challenge is to develop
(conditional or unconditional) hardness results in this model.

\bibliographystyle{plain}
\bibliography{references}



\end{document}